\documentclass[11pt]{article}
\usepackage[T1]{fontenc}
\usepackage{lmodern}
\usepackage[a4paper,margin=1in]{geometry}
\usepackage{amsmath,amssymb,amsthm,mathtools}
\usepackage{microtype}
\usepackage[hidelinks]{hyperref}
\hypersetup{
  pdftitle={A Tight Second-Order Lower Bound for Routing Labels in Trees},
  pdfauthor={Hanqing Li}
}

\newtheorem{theorem}{Theorem}
\newtheorem{lemma}[theorem]{Lemma}
\newtheorem{proposition}[theorem]{Proposition}

\newcommand{\cC}{\mathcal C}
\newcommand{\Emb}{\operatorname{Emb}}

\title{A Tight Second-Order Lower Bound for Routing Labels in Trees}
\author{Hanqing Li\\Peking University}
\date{}

\begin{document}
\maketitle

\begin{abstract}
In the designer-port routing-labeling problem, every vertex of a rooted tree
receives a binary label and the child edges receive distinct port numbers.
Given only the labels of a source and a destination, a decoder must return the
first port on their path.  Gawrychowski, Janczewski, and
\L opusza\'nski gave labels of length
$\log_2 n+O((\log_2\log_2 n)^2)$, whereas the previous lower bound was
$\log_2 n+\Omega(\log_2\log_2 n)$.  We prove that every scheme for all
$n$-vertex trees needs a label of length
$\log_2 n+\Omega((\log_2\log_2 n)^2)$ for every sufficiently large $n$.
The result allows arbitrary port assignments and imposes no computational
restriction on either the encoder or the decoder.  Thus the second-order
term in the optimal worst-case label length is determined up to constant
factors.
\end{abstract}

\section{Introduction}

Routing labels are a compact representation of the first edge on the unique
path between two vertices of a tree.  In the designer-port model the encoder
may choose the numbering of the child edges, making lower bounds especially
delicate: local port names can change from one input tree to another.
The upper bound of Gawrychowski, Janczewski, and \L opusza\'nski
\cite{GawrychowskiJanczewskiLopuszanski2021} is
$\log n+O((\log\log n)^2)$ bits.  Their lower bound is
$\log n+\Omega(\log\log n)$, and they explicitly ask whether the squared
term is necessary.  The earlier construction of Thorup and Zwick
\cite{ThorupZwick2001} established labels of $(1+o(1))\log n$ bits.

We answer the second-order question.  All logarithms below have base two.
The constants implicit in our lower bound are absolute, though very small.

\begin{theorem}\label{thm:main}
There are constants $\varepsilon>0$ and $n_0$ such that, for every
$n\ge n_0$, any designer-port routing-labeling scheme for all rooted
$n$-vertex trees uses, on some tree and vertex, at least
\[
  \log n+\varepsilon(\log\log n)^2-O(1)
\]
bits.  This holds without restrictions on encoding time, decoding time,
or the choice of port numbers.
\end{theorem}

Combined with the cited upper bound, this gives the following immediate
consequence.  The quantifier refers to the maximum label length over all
trees with exactly $n$ vertices.

\begin{theorem}\label{thm:tight}
The optimal worst-case label length for designer-port routing in rooted
$n$-vertex trees is $\log n+\Theta((\log\log n)^2)$.
\end{theorem}

The proof uses a fixed decoder and a fixed universe of available labels.
For each possible label at a root, the decoder partitions destination labels
among ports.  We build many child-size profiles of equal total mass whose
capacity demands are mutually incompatible.  Random hard subtrees turn a
shortfall in one port class into a negligible probability of a successful
labeling.  A second distribution amplifies this probability so that a union
bound can include all root labels and all port assignments.

\section{Label universes and two induction statements}

A scheme assigns each vertex a finite binary string and assigns distinct
ports $1,\ldots,d$ to the $d$ children of each vertex.  On distinct vertices
$u,w$, the decoder $D(\ell(u),\ell(w))$ must output $0$ if the first edge
goes to the parent of $u$, and otherwise the child's port.  The decoder
has no access to the tree.  For a finite set $S$ of strings, say that a
rooted tree \emph{embeds in $S$} if it has a correct routing labeling using
only labels in $S$, with arbitrary legal choices of child ports.  Write
$\Emb_D(T,S)$ for this event when $T$ is random.

\begin{lemma}\label{lem:interface}
In any correctly labeled tree, vertex labels are pairwise distinct.  For
a fixed root label $r\in S$, define
\[
 S_p(r)=\{x\in S:D(r,x)=p\},\qquad p=1,2,\ldots .
\]
The classes $S_p(r)$ are disjoint.  If a child subtree is assigned port
$p$, every label in that subtree belongs to $S_p(r)$.
\end{lemma}
\begin{proof}
If an ancestor and a proper descendant had the same label, the same
ordered pair of labels would have to produce a downward port in one
direction and $0$ in the reverse direction.  If two incomparable vertices
had the same label, their least common ancestor would have to send the
same destination label through two different ports.  Thus labels are
distinct.  The remaining assertions follow by querying the root against
each vertex of a child subtree.
\end{proof}

Fix positive constants in the order
\[
  0<\kappa\ll 1,\qquad 0<\lambda\ll\kappa,
  \qquad 0<\varepsilon\ll\lambda.
\]
Their precise inequalities are specified by the estimates below.  For
$s\ge1$ put
\begin{align*}
 m_s&=\lfloor\log s\rfloor, &
 F_m&=2^{\varepsilon(\log(m+2))^2},\\
 \xi_m&=\min\left\{\frac14,
        \frac{\lambda(\log(m+2))^2}{m+1}\right\}, &
 A_s&=\theta sF_{m_s},\qquad d(s)=(1-\xi_{m_s})A_s.
\end{align*}
Here $\theta>0$ is chosen last.  For fixed constants and all sufficiently
large $m$, $F_m$ is increasing, $\xi_m$ is decreasing, and hence $d(s)$ is
strictly increasing in $s$.  In particular, within a dyadic size level,
$d(s)$ is a constant multiple of $s$.

\begin{proposition}\label{prop:induction}
For every positive integer $s$ there are distributions $W_s$ and $R_s$
on rooted trees with exactly $s$ vertices, independent of $D$ and $S$,
such that for every decoder $D$ and every finite set $S$:
\begin{align}
 |S|\le A_s&\implies
   \Pr[\Emb_D(W_s,S)]\le s^{-10},\label{eq:weak}\\
 |S|\le d(s)&\implies
   \Pr[\Emb_D(R_s,S)]\le e^{-m_s^2}.\label{eq:strong}
\end{align}
\end{proposition}

We prove both statements by induction on $s$.  Choose a sufficiently large
absolute threshold $m_0$ after $\kappa,\lambda,\varepsilon$; then choose
$\theta$ so small that $A_s<s$ for all $s<2^{m_0}$.  For these sizes,
both conclusions hold with any distributions by
Lemma~\ref{lem:interface}.  For each larger $s$, construct $R_s$ from
smaller $W_b$ and then $W_s$ from smaller $R_b$.  This ordering avoids
any circular appeal.

\section{Amplifying the embedding failure}

Let $s\ge2^{m_0}$ and $m=m_s$.  Set $k=\lceil m^2\rceil$.
Choose $k$ distinct positive integers $b_1,\ldots,b_k$ such that
\[
 b_i=s/k+O(k),\qquad \sum_{i=1}^k b_i=s-O(k).
\]
For example, write $s-1=kq+r$, $0\le r<k$, and take the consecutive
integers $q-\lfloor k/2\rfloor,\ldots,q-\lfloor k/2\rfloor+k-1$.
After adding one branching vertex, the unused $O(k)$ vertices form a
unary chain above it.  Attach independent $W_{b_i}$ trees to its $k$
children; the resulting tree has $s$ vertices and defines $R_s$.
All $b_i<s$ and $m_{b_i}=m-O(\log m)$.

\begin{lemma}\label{lem:amplify}
For $m_0$ sufficiently large and $\varepsilon$ sufficiently small relative
to $\lambda$, the distribution $R_s$ satisfies~\eqref{eq:strong}.
\end{lemma}
\begin{proof}
The mean value theorem applied to $\log F_m$ gives, uniformly for the
chosen child sizes,
\[
 \frac{F_{m_{b_i}}}{F_m}
 \ge 1-C\varepsilon\frac{(\log m)^2}{m}
\]
for an absolute $C$.  Since $\sum b_i=s-O(m^2)$, increasing $m_0$ and
decreasing $\varepsilon/\lambda$ ensure
\begin{equation}\label{eq:sum-weak}
 \sum_i A_{b_i}\ge (1-\xi_m/2)A_s,
 \qquad \max_i A_{b_i}\le \frac{2A_s}{k}.
\end{equation}
Fix $D$ and $|S|\le d(s)=(1-\xi_m)A_s$.  Fix a possible label
$r\in S$ for the branching vertex and a bijection between its children
and ports $1,\ldots,k$.  The corresponding $k$ port classes are disjoint.
By~\eqref{eq:sum-weak}, the sum of the positive deficits
$(A_{b_i}-|S_{p(i)}(r)|)_+$ is at least $\xi_m A_s/2$.
Each deficit is at most $2A_s/k$, so at least $\xi_m k/4$
children have $|S_{p(i)}(r)|<A_{b_i}$.

For this fixed choice, the independent weak induction statements imply
that all the deficient child subtrees embed with probability at most
\[
 \prod_{i\text{ deficient}}b_i^{-10}
 \le 2^{-\Omega(\xi_m k m)}.
\]
There are at most $|S|k!$ choices of $r$ and the bijection.
For sufficiently large $m$, $|S|<2^{2m}$,
$\log(k!)=O(m^2\log m)$, and
$\xi_m k m=\Theta(\lambda m^2(\log m)^2)$.
The latter exponent dominates both the union-bound cost and $m^2$.
Thus the total success probability is at most $e^{-m^2}$.
The unary chain above the branching vertex only adds constraints, so it
cannot increase this upper bound.
\end{proof}

\section{Equal-mass profiles with separated demands}

This section constructs a codebook $\cC$ of child-size profiles for
the weak induction at every $n\in[2^m,2^{m+1})$.  The construction
depends only on $m$, not on the decoder or label set.
Put
\[
 K=\lceil\kappa\log m\rceil,
 \qquad \ell_0=\lceil4\log m\rceil,
 \qquad \ell_1=\left\lfloor\frac{m}{K^2}\right\rfloor.
\]
For each integer $\ell\in[\ell_0,\ell_1]$, let
\[
 h_\ell=\left\lfloor\frac{m}{K\ell}\right\rfloor,
 \quad H_\ell=2^{m-\ell},
 \quad N_\ell=\left\lfloor\frac{2^\ell}{16h_\ell\ell^2}\right\rfloor,
 \quad \delta_\ell=\left\lfloor\frac{H_\ell}{32h_\ell}\right\rfloor.
\]
At level $\ell$ there are $h_\ell$ pairs of blocks.  Each block has
$N_\ell$ child positions.  For pair $r\in\{1,\ldots,h_\ell\}$,
define
\[
 a_{\ell,r}=H_\ell+8h_\ell\delta_\ell
                       +4(h_\ell-r)\delta_\ell,
 \qquad \eta_t=2N_\ell-t\quad(1\le t\le2N_\ell).
\]
A profile bit $z_{\ell,r}\in\{0,1\}$ sets the $t$th child size in
this pair to
\begin{equation}\label{eq:profile}
 s_{\ell,r,t}(z)=a_{\ell,r}+\eta_t+
 \begin{cases}
   z_{\ell,r}\delta_\ell,&1\le t\le N_\ell,\\
  -z_{\ell,r}\delta_\ell,&N_\ell<t\le2N_\ell.
 \end{cases}
\end{equation}
The two signs make the total size of every pair independent of its bit.

\begin{lemma}\label{lem:profiles}
For sufficiently large $m$, the profiles above have the following
properties.
\begin{enumerate}
\item Every child size is a positive integer in $[H_\ell,2H_\ell)$.
All positions across all levels are strictly decreasing in the displayed
order for every choice of the profile bits.
\item Let $B=\sum_{\ell=\ell_0}^{\ell_1}h_\ell$.
Then $B=\Theta(m/\kappa)$.  Choosing $\kappa$ sufficiently small,
there is a binary code $\cC\subseteq\{0,1\}^B$ of minimum Hamming
distance at least $B/4$ and size at least $2^{40m}$.
\item Let $S_\ell$ be the total child mass at level $\ell$ and
$S=\sum_\ell S_\ell$.  All these masses are independent of the codeword,
and
\[
 S_\ell=\Theta(2^m/\ell^2),\qquad
 S=\Theta(2^m/\log m)<2^m/4,\qquad
 \sum_\ell\ell S_\ell=O(2^m\log m).
\]
\item For distinct codewords $z,z'$, their sorted child-size vectors
satisfy
\begin{equation}\label{eq:L1-size}
 \sum_j|s_j(z)-s_j(z')|
 \ge c_1\kappa n\frac{(\log m)^2}{m}
\end{equation}
for an absolute constant $c_1>0$ and all
$n\in[2^m,2^{m+1})$.
\end{enumerate}
\end{lemma}
\begin{proof}
Uniformly over the stated range,
$K=\Theta(\kappa\log m)$, $h_\ell\ge K-O(1)$,
and the floor operations change the relevant products by only constant
factors.  For large $m$, the quantities inside both floors defining
$N_\ell$ and $\delta_\ell$ exceed $2$, so
\[
 N_\ell\ge\frac{2^\ell}{32h_\ell\ell^2},\qquad
 \delta_\ell\ge\frac{H_\ell}{64h_\ell},\qquad
 \frac{\delta_\ell}{N_\ell}
 \ge\frac14\,2^{m-2\ell}\ell^2>1.
\]
Here $\ell\le m/K^2=o(m)$; the same bounds give
$2N_\ell=o(H_\ell)$.  Also
$h_\ell\delta_\ell\le H_\ell/32$.
The slope $\eta_t$ makes sizes strictly decrease inside each pair.
At the boundary of consecutive pairs, the smallest possible gap is
$2\delta_\ell-2N_\ell+1>0$; the
level intervals $[H_\ell,2H_\ell)$ are disjoint.
The same inequalities put all sizes below $2H_\ell$.

The harmonic sum gives
\[
 B=\sum_{\ell=\ell_0}^{\ell_1}
       \left\lfloor\frac{m}{K\ell}\right\rfloor
 =\Theta\left(\frac{m}{K}\log\frac{\ell_1}{\ell_0}\right)
 =\Theta(m/\kappa).
\]
A greedy binary code that removes Hamming balls of radius
$\lfloor B/4\rfloor$ has at least
$2^{(1-H_2(1/4)-o(1))B}$ codewords, where $H_2$ is binary entropy.
The implicit lower constant in $B=\Theta(m/\kappa)$ permits a fixed
sufficiently small $\kappa$ for which this is at least $2^{40m}$.

Each pair has total mass
$2N_\ell a_{\ell,r}+\sum_t\eta_t$.
Since $a_{\ell,r}=\Theta(H_\ell)$ and
$h_\ell N_\ell=\Theta(2^\ell/\ell^2)$,
$S_\ell=\Theta(2^m/\ell^2)$.  Summing $1/\ell^2$ and $1/\ell$
proves the two displayed aggregate estimates; the upper bound on $S$
follows by increasing $m_0$.

Flipping one bit changes the $L_1$ size vector by exactly
$2N_\ell\delta_\ell$.  Uniformly in $\ell$,
\[
 2N_\ell\delta_\ell
 =\Omega\left(\frac{2^m}{h_\ell^2\ell^2}\right)
 =\Omega\left(\frac{2^mK^2}{m^2}\right).
\]
There are at least $B/4$ flipped bits between codewords.  Since
$K^2B/m^2=\Omega(\kappa(\log m)^2/m)$ and
$n<2^{m+1}$, this proves~\eqref{eq:L1-size}.
\end{proof}

For a target size $n\in[2^m,2^{m+1})$, add one further child of fixed
size
\begin{equation}\label{eq:heavy}
 h=n-1-S.
\end{equation}
It is larger than every profile child for large $m$, and
$m_h\in\{m-1,m\}$.  All codewords now give strictly decreasing
child-size vectors with the same number of children, total size $n-1$,
and the same total mass at every dyadic level.

\section{The weak distribution and proof of the lower bound}

We first compare the strong-distribution demands for two profiles.
For a codeword $z$, write $s_j(z)$ for its sorted child sizes, including
the fixed child~\eqref{eq:heavy}, and let
\[
  T=\sum_j d(s_j(z)),\qquad P=A_n-T.
\]
Both $T$ and $P$ are independent of $z$, since within each dyadic
level $d(s)$ is linear in $s$, every level mass is fixed, and the
extra child has fixed size.  Also $P\ge0$, since $d(s)\le\theta sF_m$
for all child sizes $s<n$.

\begin{lemma}\label{lem:separation}
After choosing $\lambda$ sufficiently small relative to $\kappa$
and $\varepsilon$ sufficiently small relative to $\lambda$,
for every distinct $z,z'\in\cC$,
\begin{equation}\label{eq:separation}
 \sum_j|d(s_j(z))-d(s_j(z'))|>2P.
\end{equation}
\end{lemma}
\begin{proof}
All profile child sizes at level $\ell$ lie in
$[2^{m-\ell},2^{m-\ell+1})$.  Uniformly for
$\ell\le m/K^2=o(m)$,
\begin{equation}\label{eq:regularity}
 \frac{F_{m-\ell}}{F_m}\ge\frac12,
 \qquad 1-\xi_{m-\ell}\ge\frac12,
 \qquad
 0\le 1-\frac{F_{m-\ell}}{F_m}
 \le C\varepsilon\frac{\ell\log m}{m},
 \qquad
 \xi_{m-\ell}\le C\lambda\frac{(\log m)^2}{m}.
\end{equation}
The first and third estimates follow from the mean value theorem
for $\log F_x=\varepsilon(\log(x+2))^2$;
the other two follow directly from $\xi_x$.
Within a level, the coefficient $d(s)/s$ is fixed and at least
$\theta F_m/4$.  Equation~\eqref{eq:L1-size} therefore gives
\begin{equation}\label{eq:demand-gap}
 \sum_j|d(s_j(z))-d(s_j(z'))|
 \ge c_2\theta\kappa nF_m\frac{(\log m)^2}{m}
\end{equation}
for an absolute $c_2>0$.  The fixed child contributes zero.

Using $n=1+h+\sum_\ell S_\ell$, we bound $P$ by summing
$\theta s(F_m-(1-\xi_{m_s})F_{m_s})$ over all children,
plus $\theta F_m$ for the root.  By~\eqref{eq:regularity} and
Lemma~\ref{lem:profiles}, the profile children contribute at most
\[
 C\theta F_m\left(
   \frac{\varepsilon\log m}{m}\sum_\ell\ell S_\ell
  +\frac{\lambda(\log m)^2}{m}\sum_\ell S_\ell
 \right)
 \le C'\theta(\varepsilon+\lambda)nF_m
             \frac{(\log m)^2}{m}.
\]
The fixed child has level $m$ or $m-1$, and obeys the same bound;
the lone root term is negligible.  Thus
\begin{equation}\label{eq:slack}
 P\le C''\theta(\varepsilon+\lambda)nF_m
               \frac{(\log m)^2}{m}.
\end{equation}
Choose $\lambda$ and then $\varepsilon$ so that
$2C''(\varepsilon+\lambda)<c_2\kappa$.
Equations~\eqref{eq:demand-gap} and~\eqref{eq:slack}
give~\eqref{eq:separation}.
\end{proof}

Define $W_n$ by choosing $z$ uniformly from $\cC$ and,
independently for each child size $s_j(z)$, sampling an
$R_{s_j(z)}$ tree.  Attach these trees below a fresh root.
Every child size is less than $n$, so the strong induction applies.

\begin{lemma}\label{lem:weak}
For sufficiently large $m_0$, $W_n$ satisfies~\eqref{eq:weak}.
\end{lemma}
\begin{proof}
Fix $D$ and $S$ with $|S|\le A_n$.  For a possible root label
$r\in S$, arrange the capacities of its actual child-port classes
in decreasing order, $c_1\ge\cdots\ge c_q$.  The number $q$ of children
is fixed over the codebook.  Arrange the corresponding demands
$d(s_j(z))$ in decreasing order.  This is their size order because
$d(s)$ is increasing.  Say that $r$ \emph{covers} $z$ if
$c_j\ge d(s_j(z))$ for every $j$.

One $r$ covers at most one codeword.  Indeed, were it to cover
$z$ and $z'$, then by Lemma~\ref{lem:separation},
\[
 \sum_j c_j\ge\sum_j\max\{d(s_j(z)),d(s_j(z'))\}
 =T+\tfrac12\sum_j|d(s_j(z))-d(s_j(z'))|
 >T+P=A_n,
\]
contrary to the disjointness of port classes in $S$.
Thus the probability that some $r\in S$ covers the random profile
is at most $|S|/|\cC|$.

Now condition on a profile not covered by a given $r$.
If every child could be matched to a port of sufficient capacity, then
sorting both lists in decreasing order would give
$c_j\ge d(s_j(z))$ for all $j$: the $j$ largest demands need $j$
capacities at least as large as the $j$th demand.
For every bijection between its children and ports, some child
$i$ is mapped to a port $p$ with
$|S_p(r)|<d(s_i)$.  If the full tree embeds in $S$ with root label
$r$, then the $i$th independently sampled child tree embeds in
$S_p(r)$.  Taking a union over at most $q^2\le n^2$ possible
child--port pairs, the strong induction statement bounds this
probability by
$n^2\exp[-(0.9m)^2]$ for large $m$; all child sizes have
$m_{s_i}\ge0.9m$.  Finally, union over at most $|S|<n^2$ possible
root labels gives
\[
 \Pr[\Emb_D(W_n,S)]
 \le \frac{|S|}{|\cC|}
       +|S|n^2e^{-(0.9m)^2}
 \le \frac{n^2}{2^{40m}}+n^4e^{-(0.9m)^2}
 \le n^{-10}.
\]
The union over child--port pairs implicitly allows every adaptive
choice of port numbering.
\end{proof}

Lemmas~\ref{lem:amplify} and~\ref{lem:weak} complete the induction
in Proposition~\ref{prop:induction}.  To prove
Theorem~\ref{thm:main}, suppose all labels of a scheme for
$n$-vertex trees have length at most $L$.  All of them then belong
to the same universe of at most $2^{L+1}-1$ binary strings.
Proposition~\ref{prop:induction} says that if this universe has
at most $A_n=\theta nF_{m_n}$ strings, some $n$-vertex tree cannot
be labeled correctly, regardless of its port assignment.
Consequently
\[
 L\ge \log n+\log F_{m_n}+\log\theta-O(1)
   =\log n+\varepsilon(\log\log n)^2-O(1).
\]
This proves Theorem~\ref{thm:main}; Theorem~\ref{thm:tight}
follows from the known upper bound
\cite{GawrychowskiJanczewskiLopuszanski2021}.

\paragraph{Scope.}
The proof is nonconstructive at the level of a single explicit hard
tree and gives no lower bound on query time.  Its quantifiers cover
even decoders chosen separately for each $n$, because the two
distributional induction statements hold for each fixed decoder.

\bibliographystyle{plain}
\bibliography{references}
\end{document}